\documentclass[preprint,10pt]{elsarticle_mod}
\usepackage{url}
\usepackage{amsmath}
\usepackage{amssymb}
\usepackage{amsthm}
\usepackage{graphics}
\usepackage{graphicx}
\usepackage{times}
\usepackage{algorithm,algorithmic}

\newcommand{\EBeta}[2]{\mathrm{B}\left(#1,#2\right)}

\newtheorem{theorem}{Theorem}

\newtheorem{Lemma}[theorem]{Lemma}

\urldef{\mailsa}\path|rafal.kapelko@pwr.edu.pl|

\begin{document}

\begin{frontmatter}

\title{On The Moment Distance Between Sensors and Anchor Points\tnoteref{t1}}
\tnotetext[t1]{This is a sequel paper to \cite{kapelkokranakisIPL}, which appeared and ranked $1$st on the Top $25$ for the journal Information and Processing Letters for a period of 18 months.}
\author[pwr]{Rafa\l{} Kapelko\corref{cor1}}
\ead{rafal.kapelko@pwr.edu.pl}
\fntext[pwrfootnote]{Research supported by grant nr 0401/0017/17.}
\cortext[cor1]{Corresponding author at: Department of Computer Science,
Faculty of Fundamental Problems of Technology, Wroc{\l}aw University of Science and Technology, }
\address[pwr]{ Department of Computer Science, Faculty of Fundamental Problems of Technology, Wroc{\l}aw University of Science and Technology, Poland}
\begin{abstract}
The present paper contains additional asymptotic result over an earlier investigation of Kapelko and Kranakis.

Consider $n$ mobile sensors placed independently at random with the uniform 
distribution on the unit interval $[0,1]$. 
Fix $a$  an odd natural number. 
Let $X_i$ be the the $i-$th closest sensor to $0$ on the interval $[0,1].$ Then  the following identity holds
$$\sum_{i=1}^n\mathbf{E}\left[\left|X_i-\left(\frac{i}{n}-\frac{1}{2n}\right)\right|^a\right]=\frac{\Gamma\left(\frac{a}{2}+1\right)}{2^{\frac{a}{2}}(1+a)}\frac{1}{n^{\frac{a}{2}-1}}+O\left(\frac{1}{n^{\frac{a-1}{2}}}\right),$$
where $\Gamma(z)$ is the Gamma function.

The result this paper with earlier investigation gives as full asymptotics results of 
expected sum of displacements to the power $a$ of all sensors to move from their current location to the anchor location
$\frac{i}{n}-\frac{1}{2n},$ for $i=1,2,\dots,n,$
when $a$ is natural number.
\end{abstract}
\begin{keyword}
analysis of algorithms, 
displacement, 
distance, 
random
\end{keyword}
\end{frontmatter}


\section{Motivation}
Consider $n$ mobile sensors placed independently at random with the uniform distribution on the unit interval $[0,1]$. 
The sensors are equipped with omnidirectional sensing antennas of identical range $r=\frac{1}{2n}$; thus a sensor placed at location $x$ in the unit interval 
can sense any point at distance at most $\frac{1}{2n}$ either to the left or right of $x$.
We are interested in moving the sensors from their initial positions to new locations so as to ensure that
the unit interval is covered, i.e., every point in the unit interval is within the range of a sensor
and the cost for displacement of the sensors is minimized.
Observe that the only way to attain the coverage is for the sensors to occupy the anchor location
$t_i=\frac{i}{n}-\frac{1}{2n},$ for $i=1,2,\dots ,n$ (see Algorithm \ref{alg_1}).
\begin{algorithm}
\caption{Moving of sensors}
\label{alg_1}
\begin{algorithmic}[1]
 \REQUIRE The initial location $X_1\le X_2\le\dots X_n$ of the $n$ sensors on the unit interval $[0,1].$
 \ENSURE  The final positions of the sensors are at the anchor points $\left(\frac{i}{n}-\frac{1}{2n}\right),\,\,$ $1\le i\le n.$
 \FOR{$i=1$  \TO $n$ } 
 \STATE{move the sensor $X_{i}$ at position $\left(\frac{i}{n}-\frac{1}{2n}\right)$}
 \ENDFOR
\end{algorithmic}
\end{algorithm}

Suppose the displacement of the $i-$th sensor is a distance
$d_i$, for $i=1,2,\ldots ,n$, 
Then the cost measure for the displacement of the whole system of $n$ sensors is $\sum_{i=1}^n d_i^a,$
for some constant $a>0.$ Motivation for this cost metric arises from the fact that there are obstacles in the environment
which can obstruct the sensor movement from their initial to their final destinations.

The Algorithm \ref{alg_1} was analysed for $a=1$ in \cite{spa_2013}.
In this paper the following result was proved.
The expected sum of displacement of all $n$ sensors to move from their current location to the equidistant anchor locations is in $\Theta(\sqrt{n}).$

Then, the paper \cite{kapelkokranakisIPL} contains the following precise asymptotic result.
\begin{theorem}[cf. \cite{kapelkokranakisIPL}]
\label{thm:IPL} 
The expected sum of displacement to the power $a$ of
Algorithm \ref{alg_1} is $$\frac{\left(\frac{a}{2}\right)!}{2^{\frac{a}{2}}(1+a)}\frac{1}{n^{\frac{a}{2}-1}}+
O\left(\frac{1}{n^{\frac{a}{2}}}\right),$$  when $a$ is even natural number.
\end{theorem}
In this paper we give an exact asymptotic on the expected minimum sum of displacement to the power $a,$ when $a$ is an odd natural number
of all sensors to move from their current random location on the unit interval to the anchor location
(see Theorem \ref{thm:mainexact_oddtwo1}).

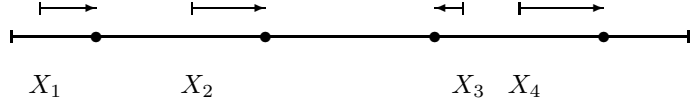
\begin{figure}[h]
\setlength{\unitlength}{0.75mm}
\centering
\begin{picture}(60,20)
\put(-30,10){\line(1,0){120}}
\put(-30,10){\line(0,1){1}}
\put(-30,10){\line(0,-1){1}}
\put(90,10){\line(0,1){1}}
\put(90,10){\line(0,-1){1}}
\put(-15,10){\circle*{2}}
\put(15,10){\circle*{2}}
\put(45,10){\circle*{2}}
\put(75,10){\circle*{2}}
\put(-27,0){$X_{1}$}
\put(-25,15){\line(0,1){1}}
\put(-25,15){\line(0,-1){1}}
\put(-25,15){\line(1,0){10}}
\put(-16,15){\vector(1,0){1}}
\put(0,0){$X_{2}$}
\put(2,15){\line(0,1){1}}
\put(2,15){\line(0,-1){1}}
\put(2,15){\line(1,0){12}}
\put(14,15){\vector(1,0){1}}
\put(48,0){$X_{3}$}
\put(50,15){\line(0,1){1}}
\put(50,15){\line(0,-1){1}}
\put(45,15){\line(1,0){5}}
\put(46,15){\vector(-1,0){1}}
\put(58,0){$X_{4}$}
\put(60,15){\line(0,1){1}}
\put(60,15){\line(0,-1){1}}
\put(60,15){\line(1,0){14}}
\put(74,15){\vector(1,0){1}}
\end{picture}
\caption{Four mobile sensors $X_1, X_2, X_3, X_4$ located on the unit interval move to the anchor points according to Algorithm \ref{alg_1}.}
\end{figure}
It is worthwhile to mention that, even though there is the closed formula on the expected sum of displacement to the power $a,$
when $a$ is an odd natural number,
the analysis of asymptotic is combinatorially challenging.

As another point of motivation for this sequel paper, we discover that the proof of the main result is a nontrivial extension.
Let us recall that the expected sum of displacement to the power $a$ of
Algorithm \ref{alg_1} is equal to
$$
\sum_{i=1}^ni\binom{n}{i}\int_{0}^{1}(t_i-x)^a x^{i-1}(1-x)^{n-i}dx,\,\,\,\text{when}\,\, a\,\, \text{is even natural number}.
$$
Then the proof of Theorem \ref{thm:IPL} is reduce to the analogue of Lemma \ref{lem:sum1}, when $a$ is even natural number
and some known identities.

Let $a$ be an odd natural number. In the paper we make an important observation 
(in the proof of Theorem \ref{thm:mainexact_oddtwo1})
that expected sum of displacement to the power $a$ of
Algorithm \ref{alg_1} is equal to
$$
\sum_{i=1}^n i\binom{n}{i}\int_{0}^{1}(x-t_i)^a x^{i-1}(1-x)^{n-i}dx
$$
$$
+\sum_{i=1}^n  2i\binom{n}{i}\int_{0}^{t_i}(t_i-x)^a x^{i-1}(1-x)^{n-i}dx.
$$
The first sum is negligible (see Lemma \ref{lem:sum1}). Thus the asymptotics depend on the expression given by the second sum.
Deriving the exact asymptotic of the second sum is not easy.
We need the sequence of technical lemmas (see Lemma \ref{lem:technical}--\ref{lem:technical3}) to prove the main result.

This sequel paper completes the previous paper \cite{kapelkokranakisIPL} and gives as full asymptotics results of 
expected sum of displacements to the power $a,$ when $a$ is natural number.
\section{Preliminaries}
In this subsection we introduce some basic concepts and notations
that will be used through the paper.

We will use many times the notations
for the rising and falling factorial \cite{concrete_1994}
$$n^{\overline{k}} = \begin{cases} 1 &\mbox{for } k=0 \\
n(n+1)\dots(n+k-1) & \mbox{for } k\ge 1, \end{cases}$$
$$n^{\underline{k}} = \begin{cases} 1 &\mbox{for } k=0 \\
n(n-1)\dots(n-(k-1)) & \mbox{for } k\ge 1. \end{cases}$$
Let $d,f$ be non-negative integers. Observe that
$$(i-1)^{\underline{d}}\cdot i^{\overline{f}}=
\frac{(i-1)^{\underline{d}}\cdot i^{\overline{f+1}}- (i-2)^{\underline{d}}\cdot (i-1)^{\overline{f+1}}}{f+d+1}.$$
Hence, by telescoping we derive
\begin{equation}
\label{eq:identity}
\sum_{i=1}^n(i-1)^{\underline{d}}\cdot i^{\overline{f}} =\frac{1}{f+d+1}(n-1)^{\underline{d}}\cdot n^{\overline{f+1}}.
\end{equation}
Notice that
\begin{equation}
\label{eq:identitysum}
\sum_{k=1}^n k^f =\frac{1}{f+1}n^{f+1}+\sum_{l=0}^{f}c_ln^l,
\end{equation}
where $c_l$ are some constans independent on $n$ (see \cite[Formula (6.78)]{concrete_1994}).

Let ${ n\brack k},$  ${n\brace k}$ be the Stirling numbers of the first and second kind respectively, which are
defined for all integer numbers such that $0\le k \le n.$ 

The following three basic equations involving Stirling numbers of the first and second kind, rising and falling factorial will be used in the proofs
(see   
\cite[Identity 6.11]{concrete_1994},
\cite[Identity 6.13]{concrete_1994}
and \cite[Identity 6.10]{concrete_1994}):
\begin{equation}
\label{eq:stirling3}
x^{\overline{m}}=\sum_{l_2}{ m\brack l_2}x^{l_2},
\end{equation}
\begin{equation}
\label{eq:stirling2}
x^{\underline{m}}=\sum_{l}{ m\brack l}(-1)^{m-l}x^{l},
\end{equation}
\begin{equation}
\label{eq:stirling}
x^m=\sum_{l}{m\brace l}x^{\underline{l}}.
\end{equation}
Assume that $b$ be is a constant independent of $m.$ Then the following Stirling numbers
\begin{equation}
\label{eq:stirling_since} 
{ m\brack m-b},\,\,\,{ m+b\brack m}\,\,\,{ m\brace m-b},\,\,\,{ m+b\brace m}\,\,\,
\end{equation}
are polynomials in the variable $n$ and of degree $2b$ (see \cite{concrete_1994}).

Let $\left\langle\left\langle n\atop k\right\rangle\right\rangle$   be the Eulerian numbers, which are
defined for all integer numbers such that $0\le k \le n.$ 
The following three identities for Euler numbers will be helpful in the deriving exact asymptotic results
(see Identities $(6.43),$ $(6.44)$ and $(6.42)$ in \cite{concrete_1994}):
\begin{equation}
\label{eq:euler2}
{m\brace m-b}=\sum_{l}\left\langle\left\langle b\atop l\right\rangle\right\rangle\binom{m+b-1-l}{2b},
\end{equation}
\begin{equation}
\label{eq:euler3}
{m\brack m-b}=\sum_{l}\left\langle\left\langle b\atop l\right\rangle\right\rangle\binom{m+l}{2b},
\end{equation}
\begin{equation}
\label{eq:euler1}
\sum_{l}\left\langle\left\langle m\atop l\right\rangle\right\rangle=\frac{(2m)!}{(m)!}\frac{1}{2^m}.
\end{equation}
We will use Euler's Finite Difference Theorem (see \cite[Identity 6.16]{Gould}). \footnote{We note that Equation (\ref{eq:triangle}) is crucial to derive the main result.}  
Assume that $a$ is a natural number. Let $f(j)=j^m$ and $m\in \mathbb{N}.$  Then
\begin{equation}
\label{eq:triangle}
\sum_{j=0}^{a}\binom{a}{j}(-1)^jf(j)=\begin{cases} 0 &\mbox{if } m <a \\
(-1)^{a}a! & \mbox{if } m= a. \end{cases}
\end{equation}
The Gamma function  (see \cite{NIST}) is defined as a definite integral for $z>0$
$$
\Gamma(z)=\int_{0}^{\infty}t^{z-1}e^{-t}dt.
$$
It is related to factorial for $n\in\mathbb{N}$
$$\Gamma(n+1)=n!.$$ 

The Beta function (see \cite{NIST}) is defined as a improper integral for positive numbers $c,d$ 
\begin{equation}
\label{eq:beta}
\EBeta{c}{d} = \int_0^1 x^{c-1}(1-x)^{d-1} dx.
\end{equation}
It is related to binomial for positive integer numbers $c,d$
\begin{equation}
\label{Emult}
\EBeta{c}{d}=\frac{1}{\binom{c+d-1}{c}c}.
\end{equation}
The Beta distribution (see \cite{NIST}) with parameters $c,d>0$  is the continuous distribution on $[0,1].$ 
The probability density function $P(x)$ and distribution function  $I(z;c,d)$ are given by
\begin{equation}
\label{eq:prosto}
P(x)=\frac{x^{c-1}(1-x)^{d-1}}{\EBeta{c}{d}},
\end{equation}
\begin{equation}
\label{incomplete:first}
I(z;c,d)=\frac{1}{\EBeta{c}{d}}\int_{0}^{z}x^{c-1}(1-x)^{d-1}dx.
\end{equation}
Notice that
\begin{equation}
\label{probal_eq}
I(z;c,d)\le 1,\,\,\text{when}\,\,\,\,c,d> 0\,\,\text{and}\,\,0\le z\le 1.
\end{equation}
Using integration by parts we can derive the following identity 
\begin{equation}
\label{incomplete:req}
I(z;c,d)=I(z;c-1,d)-\frac{\Gamma(c+d-1)}{\Gamma(c)\Gamma(d)}z^{c-1}(1-z)^d.
\end{equation}
(see also \cite[Identity 8.17.20]{NIST}).

The following binomial identity (see \cite[Identity 7.11]{Gould})
\begin{equation}
\label{eq:gould100}
\sum_{b=0}^{a}\binom{\frac{a-1}{2}}{b}(-1)^b\frac{1}{2b+1}=\frac{\sqrt{\pi}\left(\frac{a-1}{2}\right)!}{2\Gamma\left(\frac{a}{2}+1\right)}
\,\,\,\text{if}\,\,\, a \equiv 1  \pmod{2}.
\end{equation}
will be used in the proof of technical Lemma \ref{lem:technical2}.

We will use tow forms of Stirling's formula (see \cite[page 54]{feller1968}, \cite[Formula 9.40]{concrete_1994}) in the asymptotic
analysis of the sum (\ref{eq:szp}) in Lemma \ref{lem:technical3}:
\begin{equation}
\label{eq:stirlingform}
\sqrt{2\pi}m^{m+\frac{1}{2}}e^{-m+\frac{1}{12m+1}}< m!< \sqrt{2\pi}m^{m+\frac{1}{2}}e^{-m+\frac{1}{12m}},
\end{equation}
\begin{equation}
\label{eq:stirling2form}
m!= \sqrt{2\pi}m^{m+\frac{1}{2}}e^{-m}\left(1+O\left(\frac{1}{m}\right)\right).
\end{equation}
\section{Main result}
Fix $a$ an odd natural number.
This section is devoted to asymptotic analysis of Algorithm \ref{alg_1}, when
the cost of displacement of the $i-$th sensor to the anchor point $t_i=\frac{i}{n}-\frac{1}{2n},$ raised to the $a-$th power.
We prove the following theorem.
\begin{theorem} 
\label{thm:mainexact_oddtwo1} 
Let $a$ be an odd natural number. Assume that $n$ mobile sensors are thrown uniformly and independently at random in the unit interval. The expected sum
over all sensors $i$ from $1$ to $n,$ where the contribution of the $i-$th sensor is its displacement from the current location to the anchor point
$t_i=\frac{i}{n}-\frac{1}{2n},$ raised to the $a-$th power
is
$$\frac{\Gamma\left(\frac{a}{2}+1\right)}{2^{\frac{a}{2}}(1+a)}\frac{1}{n^{\frac{a}{2}-1}}+O\left(\frac{1}{n^{\frac{a-1}{2}}}\right).$$
\end{theorem}

We begin with the following sequences of lemmas which will be helpful in the proof of Theorem \ref{thm:mainexact_oddtwo1}.

First we prove the following lemma which is the analogue of lemma $1$ in \cite{kapelkokranakisIPL}. We need only the upper bound and the proof is simpler. 
The strategy of our proof is the following. We rewrite the inner sum 
$\sum_{i=1}^{n}\frac{1}{n^a}\binom{a}{j}(-1)^jn^j  \frac{ \left(i-\frac{1}{2}\right)^{a-j}\cdot i^{\overline{j}}}{(n+1)^{\overline{j}}}$
as the sum $\sum_{k=0}^{a+1}C_{a+1-k}\,\, n^{a+1-k},$ where $C_{a+1-k}$ is the polynomial of variable $j$ of degree less than or equal $2k.$
Then, we apply Euler's finite difference theorem to prove Equation (\ref{eq:desired100}),
giving the desired upper bound.
\begin{Lemma}
\label{lem:sum1}
Assume that $a$ is an odd natural number. Then
$$
\sum_{j=0}^{a}\sum_{i=1}^{n}\frac{1}{n^a}\binom{a}{j}(-1)^jn^j  \frac{ \left(i-\frac{1}{2}\right)^{a-j}\cdot i^{\overline{j}}}{(n+1)^{\overline{j}}}=
O\left(\frac{1}{n^{\frac{a-1}{2}}}\right). 
$$
\end{Lemma}
\begin{proof}
Let $j\in\{0,\dots, a\}.$
Applying Equation (\ref{eq:stirling}), Equation (\ref{eq:stirling2})  as well Equation (\ref{eq:identity}) we deduce that
\begin{align*}
&\sum_{i=1}^{n} n^j  \frac{ \left(i-\frac{1}{2}\right)^{a-j}\cdot i^{\overline{j}}}{(n+1)^{\overline{j}}}\\
&=
n^j\sum_{l_1}\binom{a-j}{l_1}\left(\sum_{i=1}^{n}\frac{(i-1)^{a-j-l_1}\left(\frac{1}{2}\right)^{l_1}i^{\overline{j}}}{(n+1)^{\overline{j}}}\right)\\
&=
n^j\sum_{l_1}\binom{a-j}{a-j-l_1}\left(\frac{1}{2}\right)^{l_1}\sum_{l_2}{a-j-l_1\brace l_2}  \left(\sum_{i=1}^{n}\frac{(i-1)^{\underline{l_2}}i^{\overline{j}}}{(n+1)^{\overline{j}}}\right)\\
&=n^j\sum_{l_1}\binom{a-j}{a-j-l_1}\left(\frac{1}{2}\right)^{l_1}\sum_{l_2}{a-j-l_1\brace l_2} \frac{1}{l_2+j+1}n^{\underline{l_2+1}}\\
&=\sum_{l_1}\binom{a-j}{a-j-l_1}\left(\frac{1}{2}\right)^{l_1}\\
&\times\left(\sum_{l_2}\sum_{l_3}{a-j-l_1\brace l_2}{l_2+1\brack l_3}(-1)^{l_2+1-l_3} \frac{1}{l_2+j+1}n^{l_3+j}\right).
\end{align*}
Hence $$\sum_{i=1}^{n} n^j  \frac{ \left(i-\frac{1}{2}\right)^{a-j}\cdot i^{\overline{j}}}{(n+1)^{\overline{j}}}=\sum_{k=0}^{a+1}C_{a+1-k}\,\, n^{a+1-k}$$ and
\begin{align*}
C&_{a+1-k}\\
&=\sum_{l_1}\binom{a-j}{a-j-l_1}\left(\frac{1}{2}\right)^{l_1}\\
&\times\sum_{l_2}{a-j-l_1\brace l_2}{l_2+1\brack a+1-k-j}(-1)^{l_2-a+k+j} \frac{1}{l_2+j+1}\\
&=\sum_{l_1}\binom{a-j}{a-j-l_1}\left(\frac{1}{2}\right)^{l_1}\\
&\times\sum_{l_4}{a-j-l_1\brace a-j+l_4-k}{a-j+l_4-k+1\brack a+1-k-j}\frac{(-1)^{l_4}}{a+1-k+l_4}.
\end{align*}
Since $\binom{a-j}{a-j-l_1}$ is the polynomial of variable $a-j$ of degree $l_1,$
${a-j-l_1\brace a-j+l_4-k}$ is the polynomial of variable $a-j-l_1$ of degree $2(k-l_1-l_4)$ (see (\ref{eq:stirling_since}))
and ${a-j+l_4-k+1\brack a+1-k-j}$ is the polynomial of variable $a+1-j-k$ of degree $2l_4$ (see (\ref{eq:stirling_since})),
we obtain that, $C_{a+1-k}$ is the polynomial of variable $a-j$ of degree less than or equal $2k.$ 
Hence from Equation (\ref{eq:triangle}) we get
\begin{equation}
\label{eq:desired100}
\sum_{j=0}^{a}\binom{a}{j}(-1)^jC_{a+1-k}= 0\,\, \text{for }\,\, 2k <a. 
\end{equation}
Putting everything together, we obtain
\begin{align*}
\sum_{j=0}^{a}&\sum_{i=1}^{n}\frac{1}{n^a}\binom{a}{j}(-1)^jn^j  \frac{ \left(i-\frac{1}{2}\right)^{a-j}\cdot i^{\overline{j}}}{(n+1)^{\overline{j}}}\\
=&\sum_{k=0}^{a+1}\frac{n^{a+1-k}}{n^a}\sum_{j=0}^{a}\binom{a}{j}(-1)^j\, C_{a+1-k}\\
=&\sum_{k=\frac{a+1}{2}}^{a+1}\frac{n^{a+1-k}}{n^a}\sum_{j=0}^{a}\binom{a}{j}(-1)^j\, C_{a+1-k}\\
=& O\left(\frac{1}{n^{\frac{a-1}{2}}}\right).
\end{align*}
This completes the proof of Lemma \ref{lem:sum1}.
\end{proof}
We will prove the following technical Lemma \ref{lem:technical}. 

The general strategy of our proof is to rewrite the expression 
$n^j\left(i-\frac{1}{2}\right)^{a-j}i^{\overline{j}}(n+a)^{\underline{a-j}}$
as the sum
$\sum_{0\le l_1,l_2\le a}a_{l_1, l_2}(j)n^{a-l_1}i^{a-l_2},$
where $a_{l_1, l_2}(j)$ are the polynomials of variable $j$ of degree less than or equal $2l_1+2l_2.$
Then, we combine Euler's finite difference theorem, the probabilistic inequality (\ref{eq:02star}) and Equation
(\ref{eq:identity}) and get the desired upper bound.
\begin{Lemma}
\label{lem:technical}
Assume that $a$ is an odd natural number. Let $t_i=\frac{i}{n}-\frac{1}{2n}$ and let
$$A_{i}^{(a)}=\frac{1}{n^a(n+1)^{(a)}}\sum_{j=0}^{a}\binom{a}{j}(-1)^jn^j\left(i-\frac{1}{2}\right)^{a-j}i^{\overline{j}}(n+a)^{\underline{a-j}}.$$ Then
$$\sum_{i=1}^{n}A_{i}^{(a)}\binom{n}{i}i\int_{0}^{t_i}x^{i-1}(1-x)^{n-i}dx=O\left(\frac{1}{n^{\frac{a-1}{2}}}\right).$$ 
\end{Lemma}
\begin{proof}  
Observe that  $$A_{i}^{(a)}=\frac{1}{n^a(n+1)^{(a)}}\sum_{0\le l_1,l_2\le a}\sum_{j=0}^{a}\binom{a}{j}(-1)^j a_{l_1, l_2}(j)n^{a-l_1}i^{a-l_2}.$$ 
Applying Equation (\ref{eq:stirling2}) and Equation (\ref{eq:stirling3}) we deduce that
\begin{align*}
n^j&\left(i-\frac{1}{2}\right)^{a-j}i^{\overline{j}}(n+a)^{\underline{a-j}}\\
&=n^j\sum_{l_5}\binom{a-j}{l_5}i^{\overline{a-j-l_5}}\left(-\frac{1}{2}\right)^{l_5}i^{\overline{j}}(n+a)^{\underline{a-j}}\\
&=\left(\sum_{l_5}\binom{a-j}{l_5}\left(-\frac{1}{2}\right)^{l_5}\left(\sum_{l_6}i^{a-j-l_5}{ j\brack j-l_6}i^{j-l_6}\right)\right)\\
&\times\left(\sum_{l_3}n^j{a-j\brack a-j-l_3}(-1)^{l_3}(n+a)^{a-j-l_3}\right)\\
&=\left(\sum_{l_5}\binom{a-j}{l_5}\left(-\frac{1}{2}\right)^{l_5}\left(\sum_{l_6}i^{a-l_5-l_6}{ j\brack j-l_6}\right)\right)\\
&\times\left(\sum_{l_3}{a-j\brack a-j-l_3}(-1)^{l_3}\left(\sum_{l_4}n^{a-l_3-l_4}a^{l_4}\binom{a-j-l_3}{l_4}\right)\right).
\end{align*}
Therefore, the coefficient of  the term $i^{a-l_1}n^{a-l_2}$ in the polynomial \\
$n^j\left(i-\frac{1}{2}\right)^{a-j}i^{\overline{j}}(n+a)^{\underline{a-j}}$ equals
\begin{align*}
&a_{l_1, l_2}(j)=\left(\sum_{l_5+l_6=l_1}\binom{a-j}{l_5}\left(-\frac{1}{2}\right)^{l_5}{ j\brack j-l_6}\right)\\
&\times\left(\sum_{l_3+l_4=l_2}{a-j\brack a-j-l_3}(-1)^{l_3}a^{l_4}\binom{a-j-l_3}{l_4}\right).
\end{align*}
Since ${ j\brack j-l_6}$ is the polynomial of variable $j$ of degree $2l_6$, $\binom{a-j}{l_5}$ is the polynomial of variable $a-j$ of degree $l_5$, ${a- j\brack a- j-l_3}$
is the polynomial of variable $a-j$ of degree $2l_3$ (see (\ref{eq:stirling_since})) and $\binom{a-j-l_3}{l_4}$ is the polynomial of variable $a-j-l_3$ of degree $l_4,$
we obtain that, $a_{l_1, l_2}(j)$ is the polynomial of variable $j$ of degree less than or equal $2l_1+2l_2.$ 
Hence from Equation (\ref{eq:triangle}) we get 
$$\sum_{j=0}^a\binom{a}{j}(-1)^ja_{l_1, l_2}(j)=0\,\,\, \text{for}\,\,\, 2l_1+2l_2<a+1.$$
Therefore
\begin{equation}
\label{eq:01star}
A_{i}^{(a)}=\frac{1}{n^a(n+1)^{\overline{a}}}\sum_{{S_1}}
\left(\sum_{j=0}^{a}\binom{a}{j}(-1)^j a_{l_1, l_2}(j)\right)n^{a-l_1}i^{a-l_2},
\end{equation}
where ${S_1}=\left\{(l_1,l_2): 0\le l_1,l_2\le a,\,\, \frac{a+1}{2}\le l_1+l_2\right\}.$ 
On the other hand, from Equation (\ref{probal_eq}) for $c=i,$ $d=n-i+1$ and $z=t_i$ we deduce that
\begin{equation}
\label{eq:02star}
\binom{n}{i}i\int_0^{t_i} x^{i-1}(1-x)^{n-i}dx\le 1.
\end{equation}
Putting together Equation (\ref{eq:01star}), Equation (\ref{eq:02star}) and Identity (\ref{eq:identitysum}) for $k=i$ and $f=a-l_2$ we have
$$\sum_{i=1}^{ n }i^{a-l_2}={n}^{a-l_2+1}\frac{1}{a-l_2+1}+O(n^{a-l_2}).$$
Therefore
$$\sum_{i=1}^{ n}A_i^{(a)}\binom{n}{i} i \int_{0}^{t_i} x^{i-1}(1-x)^{n-i}dx=O\left(\frac{1}{n^{\frac{a-1}{2}}}\right).$$
This completes the proof of Lemma \ref{lem:technical}. 
\end{proof}  
The proof of the next lemma is technically complicated. Before starting the proof, we briefly explain the overall strategy of the analysis.

\textbf{Case of Equation (\ref{eq:technical2a}).}
We write $B_i^{(a)}(n+1)^{\overline{a}}$ as the polynomial in two variables $n,i-\frac{1}{2}.$ Using the property (\ref{eq:triangle})
we show that several coefficient of the polynomial  $B_i^{(a)}(n+1)^{\overline{a}}$ are zero.

\textbf{Case of Equation (\ref{eq:technical2a}).}
Firstly, we apply Equation (\ref{eq:triangle}) and reduce $b_{q_1, p_1}(a)$ to the Sum (\ref{eq:b-q-p}).
Using the definition of the Beta function we rewrite Sum (\ref{eq:100}) as the double sum and the definite integral (see
(\ref{eq:100}) and (\ref{eq:200})).
Then again we apply Equation (\ref{eq:triangle}). This reduces the integral (\ref{eq:200})) to the integral
$\int_0^1x^{\frac{a}{2}}(1-x)^{\frac{a}{2}}dx$ and the Sum (\ref{eq:technical2a})
to the Sum (\ref{eq:300}). Then, using (\ref{eq:triangle}) and (\ref{eq:gould100})
we get the desired Equality (\ref{eq:400}).
\begin{Lemma}
\label{lem:technical2}
Assume that $a$ is an odd natural number. Let $t_i=\frac{i}{n}-\frac{1}{2n}$ and let
$$B_{i}^{(a)}=\sum_{j=1}^{a}\binom{a}{j}(-1)^{j+1}\sum_{k=1}^{j}\frac{\binom{n+k-1}{k+i-1}}{\binom{n+j}{i+j}(i+j)}(t_i)^{k-j}.$$ Then
\begin{equation}
\label{eq:technical2a}
B_{i}^{(a)}=\frac{1}{(n+1)^{\overline{a}}}\sum_{S_2
}
b_{q_1, p_1}(a)n^{a-1-q_1}\left(i-\frac{1}{2}\right)^{-p_1},
\end{equation}
where  
 ${S_2}=\left\{(q_1,p_1): 0\le q_1,p_1\le a-1,\,\, q_1+p_1\ge \frac{a-1}{2}\right\},$ 
$b_{q_1, p_1}(a)$ are some constants independent on $i, n\,\,$ and 
\begin{equation}
\label{eq:technical2b}
\sum_{q_1+p_1=\frac{a-1}{2}}\frac{2}{\sqrt{2\pi}}\EBeta{a-p_1+\frac{1}{2}}{\frac{3}{2}}b_{q_1,p_1}(a)=\frac{\Gamma\left(\frac{a}{2}+1\right)}{2^{\frac{a}{2}}(1+a)}.
\end{equation}
\end{Lemma}
\begin{proof} 
We will discuss separately proof of Equation (\ref{eq:technical2a}) and (\ref{eq:technical2b}).

\textbf{Case of Equation (\ref{eq:technical2a}).}

Define
$$C_{(i,n)}^{(a)}(j)=\sum_{k=1}^{j}\frac{\binom{n+k-1}{k+i-1}}{\binom{n+j}{i+j}(i+j)}(t_i)^{k-j}.$$
Observe that
$$C_{(i,n)}^{(a)}(j)=\frac{1}{\left(i-\frac{1}{2}\right)^{a-1}(n+1)^{\overline{a}}n}\sum_{S_3
}
a_{q_1, p_1}(j)n^{a-q_1}\left(i-\frac{1}{2}\right)^{a-1-p_1},$$
where  
${S_3}=\left\{(q_1,p_1): 0\le q_1,p_1\le a-1\right\},$ 
$a_{q_1, p_1}(j)$ are some constants independent on $i, n.\,\,$
On the other hand  
\begin{align*}
&C_{(i,n)}^{(a)}(j)\\
&=\frac{(n+a)^{\underline{a-j}}}{\left(i-\frac{1}{2}\right)^{a-1}(n+1)^{\overline{a}}n}\sum_{k=1}^{j}n^{j-k}n^{(k)}\left(i-\frac{1}{2}\right)^{a-1-k-j}
(i+j-1)^{\underline{j-k}}.
\end{align*}
Applying Equation (\ref{eq:stirling2}) and Equation (\ref{eq:stirling3}) we deduce that
\begin{align*}
&n^{j-k}n^{(k)}(n+a)_{(a-j)}\\
&=\left(\sum_{q_2}n^{-k}{ k\brack k-q_2}n^{k-q_2}\right)\\
&\times\left(\sum_{q_3}n^j{a-j\brack a-j-q_3}(-1)^{q_3}(n+a)^{a-j-q_3}\right)\\
&=\left(\sum_{q_2}{ k\brack k-q_2}n^{-q_2}\right)\\
&\times\left(\sum_{q_3}{a-j\brack a-j-q_3}(-1)^{q_3}
\left(\sum_{q_4}n^{a-q_3-q_4}a^{q_4}\binom{a-j-q_3}{q_4}\right)\right).
\end{align*}
Therefore, the coefficient of  the term $n^{a-q_1}$ in the polynomial $n^{j-k}n^{\overline{k}}(n+a)^{\underline{a-j}}$ equals
\begin{equation}
 \label{eq:c}
c^{(k)}_{q_1}(j)=\sum_{q_2+q_3+q_4=q_1}\left({ k\brack k-q_2}{a-j\brack a-j-q_3}(-1)^{q_3}a^{q_4}\binom{a-j-q_3}{q_4}\right).
\end{equation}
Applying Equation (\ref{eq:stirling2}) and Equation
(\ref{eq:stirling3}) we deduce that
\begin{align*}
&\left(i-\frac{1}{2}\right)^{a-1+k-j}(i+j-1)^{\underline{j-k}}\\
&\,\,=\left(i-\frac{1}{2}\right)^{a-1+k-j} \sum_{p_2}{j-k\brack j-k-p_2}(-1)^{p_2}(i+j-1)^{j-k-p_2}\\
&\,\,=\left(i-\frac{1}{2}\right)^{a-1+k-j}\sum_{p_2}{j-k\brack j-k-p_2}(-1)^{p_2}\\
&\,\,\times\left(\sum_{p_3}\left(i-\frac{1}{2}\right)^{j-k-p_2-p_3}\left(j-\frac{1}{2}\right)^{p_3}\binom{j-k-p_2}{p_3}\right).
\end{align*}
Therefore, the coefficient of  the term $\left(i-\frac{1}{2}\right)^{a-1-p_1}$ in the polynomial
$\left(i-\frac{1}{2}\right)^{a-1+k-j}(i+j-1)^{\underline{j-k}}$ equals
\begin{equation}
\label{eq:d} 
d^{(k)}_{p_1}(j)=\sum_{p_2+p_3=p_1}{j- k\brack j- k-p_2} \binom{j-k-p_2}{p_3}(-1)^{p_2}\left(j-\frac{1}{2}\right)^{p_3}.
\end{equation}
Hence, the coefficient of the term $n^{a-q_1}\left(i-\frac{1}{2}\right)^{a-1-p_1}$ in the polynomial 
$$\sum_{k=1}^{j}n^{j-k}n^{\overline{k}}(n+a)^{\underline{a-j}}\left(i-\frac{1}{2}\right)^{a-1+k-j}(i+j-1)^{\underline{j-k}}$$
equals
$$a_{q_1,p_1}(j)=\sum_{k=1}^{j}c^{(k)}_{q_1}(j)d^{(k)}_{p_1}(j).$$
Notice that
\begin{equation}
\label{ex:zero}
 a_{q_1,p_1}(0)=0.
\end{equation}
Observe that ${k\brack k-q_2}$ is the polynomial of variable $k$ of degree $2q_2$ (see (\ref{eq:stirling_since})), 
${ a-j\brack a-j-q_3}$ is the polynomial of variable $a-j$ of degree $2q_3$ and $\binom{a-j-q_3}{q_4}$ is the polynomial of variable $a-j-q_3$ of degree $q_4.$
Therefore $c^{(k)}_{q_1}(j)$ is the polynomial of variable $k$ of degree $2q_2$
and is the polynomial of variable $j$ of degree less than or equal $2q_3+q_4.$

Observe that
${ j-k\brack j-k-p_2}$ is the polynomial of variable $j-k$ of degree $2p_2$ (see (\ref{eq:stirling_since})) 
 and $\binom{j-k-p_2}{p_3}$ is the polynomial of variable $j-k$ of degree $p_3.$ 
Therefore, $d^{(k)}_{q_1}(j)$ is the polynomial of variable $j-k$ of degree less than or equal $2p_2+p_3$
and is the polynomial of variable $j$ of degree $p_3.$

Applying these and Identity (\ref{eq:identitysum}) for  $n=j$ 
we conclude that $a_{q_1,p_1}(j)$ is the polynomial of variable $j$ of degree less than or equal 
$2q_2+2q_3+q_4+2p_2+2p_3+1.$ From  $q_2+q_3+q_4=q_1$ and $p_2+p_3=p_1$ we have that
$a_{q_1,p_1}(j)$ is the polynomial of variable $j$ of degree less than or equal 
$2p_1+2q_1+1.$

Therefore from Equation (\ref{eq:triangle}) and Equation (\ref{ex:zero}) we deduce that
$$\sum_{j=0}^{a}\binom{a}{j}(-1)^{j+1}a_{q_1,p_1}(j)=0 \,\,\,\text{for}\,\,\, q_1+p_1<\frac{a-1}{2}.$$
Let 
\begin{equation}
 \label{eq:binomiden}
 b_{q_1,p_1}(a)=\sum_{j=0}^{a}\binom{a}{j}(-1)^{j+1}a_{q_1,p_1}(j).
\end{equation}
Therefore
$$b_{q_1,p_1}(a)=0\,\,\,\text{for}\,\,\,q_1+p_1< \frac{a-1}{2}.$$
Hence
$$B_{i}^{(a)}=\frac{1}{\left(i-\frac{1}{2}\right)^{a-1}(n+1)^{\overline{a}}}\sum_{S_4
}
b_{q_1, q_2}(a)n^{a-1-q_1}\left(i-\frac{1}{2}\right)^{a-1-p_1},$$
where
$$
S_4=\left\{(q_1,p_1): 0\le q_1,p_1\le a-1\,\,q_1+p_1\ge\frac{a-1}{2}\right\}.
$$
This is enough to prove the first part of Lemma \ref{lem:technical2}.

\textbf{Case of Equation (\ref{eq:technical2b}).}

Let us recall that
$$a_{q_1,p_1}(j)=\sum_{k=1}^{j}c^{(k)}_{q_1}(j)d^{(k)}_{p_1}(j).$$
and
$$
 b_{q_1,p_1}(a)=\sum_{j=0}^{a}\binom{a}{j}(-1)^{j+1}a_{q_1,p_1}(j).
$$
(see Equation (\ref{eq:c}) and Equation (\ref{eq:d}) for the definition of $c^{(k)}_{q_1}(j)$ and $d^{(k)}_{p_1}(j)$).\\
Applying Identity (\ref{eq:euler3}) for
$m=k,$ $b=q_2$ and Identity (\ref{eq:euler3}) for
$m=a-j,$ $b=q_3$ we observe that 
the coefficient of the term $k^{2q_2}(a-j)^{2q_3}$ in the polynomial
$c^{(k)}_{q_1}(j)$ equals
$$
A_{q_2,q_3}=\left(\sum_{j_1}\left\langle\left\langle q_2\atop j_1\right\rangle\right\rangle \frac{1}{(2q_2)!}\right)
\left(\sum_{j_2}  \left\langle\left\langle q_3\atop j_2\right\rangle\right\rangle
\frac{1}{(2(q_3))!}\right)(-1)^{q_3}.
$$
Therefore, from Equation (\ref{eq:euler1}) we have
$$
A_{q_2,q_3}=\frac{(-1)^{q_3}}{(q_2)!2^{q_2}(q_1)!2^{q_1}}.
$$
Applying Identity (\ref{eq:euler3}) for
$m=j-k,$ $b=p_2$ we observe that 
the coefficient of the term $(j-k)^{2p_2}(j-k)^{p_3}\left(j-\frac{1}{2}\right)^{p_3}$ in the polynomial
$c^{(k)}_{q_1}(j)$ equals
$$
B_{p_2,p_3}=\left(\sum_{j_1}\left\langle\left\langle p_2\atop j_1\right\rangle\right\rangle \frac{1}{(2p_2)!}\right)
\frac{1}{(p_3)!}(-1)^{p_2}.
$$
Therefore, from Equation (\ref{eq:euler1}) we have
$$
B_{p_2,p_3}=\frac{(-1)^{p_2}}{(p_2)!2^{p_2}(p_3)!}.
$$
Hence
\begin{align}
c^{(k)}_{q_1}(j)&=\sum_{q_2+q_3=q_1}A_{q_2,q_3}k^{2q_2}(a-j)^{2q_3}+\sum_{l_1+l_2<2q_1}a_{l_1,l_2}k^{l_1}(a-j)^{l_2}\nonumber\\
&\label{eq:ca}=\left(\frac{k^2}{2}-\frac{(a-j)^2}{2}\right)^{q_1}\frac{1}{(q_1)!}+\sum_{l_1+l_2<2q_1}a_{l_1,l_2}k^{l_1}(a-j)^{l_2},
\end{align}
where $a_{l_1,l_2}$ are some constans independent on $k,j,$
\begin{align}
d^{(k)}_{p_1}(j)&=\sum_{p_2+p_3=p_1}B_{p_2,p_3}(j-k)^{2p_2}(j-k)^{p_3}\left(j-\frac{1}{2}\right)^{p_3}\nonumber\\
&+\sum_{l_1+l_2<2p_1}b_{l_3,l_4}(j-k)^{l_3}(j)^{l_4}\nonumber\\
&=\left((j-k)\left(j-\frac{1}{2}\right)-\frac{(j-k)^2}{2}\right)^{p_1}\frac{1}{(p_1)!}\nonumber\\
&\label{eq:ka}+\sum_{l_3+l_4<2p_1}b_{l_3,l_4}(j-k)^{l_3}(j)^{l_4},
\end{align}
where $b_{l_3,l_4}$ are some constans independent on $k,j.$

Putting together Formula (\ref{eq:ca}), Formula (\ref{eq:ka}), Identity (\ref{eq:identitysum}) for $f=l_1$ and  $f=l_3$ as well Euler's Finite
Difference Theorem (see Identity (\ref{eq:triangle})) we obtain
\begin{align}
\nonumber b_{q_1,p_1}(a)&=\sum_{j=0}^{a}\binom{a}{j}(-1)^{j+1}\sum_{k=1}^{j}c^{(k)}_{q_1}(j)d^{(k)}_{p_1}(j)\\
\nonumber &=\sum_{j=0}^{a}\binom{a}{j}(-1)^{j+1}\sum_{k=1}^{j}\left(\frac{k^2}{2}-\frac{(a-j)^2}{2}\right)^{q_1}\frac{1}{(q_1)!}\\
\label{eq:b-q-p}&\times\left((j-k)\left(j-\frac{1}{2}\right)-\frac{(j-k)^2}{2}\right)^{p_1}\frac{1}{(p_1)!}.
\end{align}
Applying this 
we get
\begin{align}
&\sum_{q_1+p_1=\frac{a-1}{2}}\EBeta{a-p_1+\frac{1}{2}}{\frac{3}{2}} b_{q_1,p_1}(a)\nonumber\\ 
&\,\,\,=\sum_{q_1+p_1=\frac{a-1}{2}}\EBeta{a-p_1+\frac{1}{2}}{\frac{3}{2}}
\sum_{j=0}^{a}\binom{a}{j}(-1)^{j+1}\sum_{k=1}^{j}c^{(k)}_{q_1}(j)d^{(k)}_{p_1}(j)\nonumber\\
&\,\,\,=\sum_{j=0}^{a}\binom{a}{j}(-1)^{j+1}\sum_{k=1}^{j} \sum_{q_1+p_1=\frac{a-1}{2}}\EBeta{a-p_1+\frac{1}{2}}{\frac{3}{2}}\frac{1}{(q_1)!}\nonumber\\
&\,\,\,\label{eq:100}\times\left(\frac{k^2}{2}-\frac{(a-j)^2}{2}\right)^{q_1}
\left((j-k)\left(j-\frac{1}{2}\right)-\frac{(j-k)^2}{2}\right)^{p_1}\frac{1}{(p_1)!}.
\end{align}

Using the definition of the Beta function (see (\ref{eq:beta}) for $c=a-p_1+\frac{1}{2},$ $d=\frac{3}{2}$)
we have
\begin{align}
&\sum_{q_1+p_1=\frac{a-1}{2}}\EBeta{a-p_1+\frac{1}{2}}{\frac{3}{2}}\nonumber\\
\,\,&\times\left(\frac{k^2}{2}-\frac{(a-j)^2}{2}\right)^{q_1}\frac{1}{(q_1)!}
\left((j-k)\left(j-\frac{1}{2}\right)-\frac{(j-k)^2}{2}\right)^{p_1}\frac{1}{(p_1)!}\nonumber\\
\,\,&=\int_{0}^{1}\sum_{q_1+p_1=\frac{a-1}{2}}(1-x)^{1/2}x^{a-1-p_1+1/2}\left(\frac{k^2}{2}-\frac{(a-j)^2}{2}\right)^{q_1}\frac{1}{(q_1)!}\nonumber\\
\,\,&\times\left((j-k)\left(j-\frac{1}{2}\right)-\frac{(j-k)^2}{2}\right)^{p_1}\frac{1}{(p_1)!}dx\nonumber\\
\,\,&=\frac{1}{\left(\frac{a-1}{2}\right)!}\int_0^1(1-x)^{1/2}x^{a-1/2}\nonumber\\
\,\,&\times\left(\frac{(j-k)\left(j-\frac{1}{2}\right)-\frac{(j-k)^2}{2}}{x}+\frac{k^2}{2}-\frac{(a-j)^2}{2}\right)^{\frac{a-1}{2}}dx\nonumber\\
\,\,&=\label{eq:200}\frac{1}{\left(\frac{a-1}{2}\right)!}\int_0^1(1-x)^{1/2}x^{a-1/2}\nonumber\\
\,\,&\times\left((j-k)(j+k)\left(\frac{1-x}{2x}\right)+a\left(j-\frac{a}{2}\right)-\frac{j-k}{2x}\right)^{\frac{a-1}{2}}dx.
\end{align}
Observe that
\begin{align*}
&\left((j-k)(j+k)\left(\frac{1-x}{2x}\right)+a\left(j-\frac{a}{2}\right)-\frac{j-k}{2x}\right)^{\frac{a-1}{2}}\\
&\,\,\,\,\,\,\,\,\,\,=\left((j-k)(j+k)\left(\frac{1-x}{2x}\right)\right)^{\frac{a-1}{2}}+
\sum_{l_3+l_4<a-1}c_{l_5,l_6}(j)^{l_3}(k)^{l_4},
\end{align*}
where $c_{l_5,l_6}$ are some constans independent on $j,k.$

Therefore, from Equation (\ref{eq:identitysum}) for $f=l_4$  and Identity (\ref{eq:triangle}) we have
\begin{align}
&\sum_{j=0}^{a}\binom{a}{j}(-1)^{j+1}
\sum_{k=1}^{j}\left((j-k)(j+k)\left(\frac{1-x}{2x}\right)+a\left(j-\frac{a}{2}\right)-\frac{j-k}{2x}\right)^{\frac{a-1}{2}}\nonumber\\
&\label{eq:300}=\left(\frac{1-x}{2x}\right)^{\frac{a-1}{2}}\sum_{j=0}^{a}\binom{a}{j}(-1)^{j+1}\sum_{k=1}^{j}\left((j-k)(j+k)\right)^{\frac{a-1}{2}}.
\end{align}
Notice that
$$((j-k)(j+k))^{\frac{a-1}{2}}=(j^2-k^2)^{\frac{a-1}{2}}=\sum_{b=0}^{\frac{a-1}{2}}j^{a-1-2b}(-1)^bk^{2b}\binom{\frac{a-1}{2}}{b}.$$
Therefore, from Equation (\ref{eq:identitysum}) for $f=2b$  and Identity (\ref{eq:gould100})
we have
\begin{align*}
\sum_{j=0}^{a}&\binom{a}{j}(-1)^{j+1}\sum_{k=1}^{j}\left((j-k)(j+k)\right)^{\frac{a-1}{2}}\\
&=\sum_{j=0}^{a}\binom{a}{j}(-1)^{j+1} \left(\frac{\sqrt{\pi}\left(\frac{a-1}{2}\right)!}{2\Gamma(\frac{a}{2}+1)}j^a+\sum_{l_7<a}d_{l_7}j^{l_7}\right),
\end{align*}
where $d_{l_7}$ are some constans indpendent on $j.$

Applying this and Identity (\ref{eq:triangle}) we get
\begin{equation}
 \label{eq:400}
\sum_{j=0}^{a}\binom{a}{j}(-1)^{j+1}\sum_{k=1}^{j}\left((j-k)(j+k)\right)^{\frac{a-1}{2}}=a!\frac{\sqrt{\pi}\left(\frac{a-1}{2}\right)!}{2\Gamma(\frac{a}{2}+1)}.
\end{equation}
Putting together Formulas (\ref{eq:100}), (\ref{eq:200}), (\ref{eq:300}), (\ref{eq:400}) and
the definition of the Beta function (see (\ref{eq:beta}) for $c=\frac{a}{2}+1,$ $d=\frac{a}{2}+1$) 
we deduce that
\begin{align*}
\sum_{q_1+p_1=\frac{a-1}{2}}&\frac{2}{\sqrt{2\pi}}\EBeta{a-p_1+\frac{1}{2}}{\frac{3}{2}}b_{q_1,p_1}(a)\\
&=\frac{a!}{2^{\frac{a}{2}}\Gamma(\frac{a}{2}+1)}\EBeta{\frac{a}{2}+1}{\frac{a}{2}+1}.
\end{align*}
Finally, using the basic identity $\EBeta{c}{d}=\frac{\Gamma(c)\Gamma(d)}{\Gamma(c+d)}$ for $c=d=\frac{a}{2}$ (see \cite[Identity 5.12.1]{NIST}) and $\Gamma(a+2)=(a+1)!$ we get
$$
\sum_{q_1+p_1=\frac{a-1}{2}}\frac{2}{\sqrt{2\pi}}\EBeta{a-p_1+\frac{1}{2}}{\frac{3}{2}}b_{q_1,p_1}(a)= \frac{\Gamma\left(\frac{a}{2}+1\right)}{2^{\frac{a}{2}}(1+a)}. 
$$
This finishes the proof of Lemma \ref{lem:technical2}. 
\end{proof}

We are now ready to give the precise asymptotic in Lemma \ref{lem:technical3}. 

Before proving Lemma \ref{lem:technical3}, let us notice that there is some interest in research community for finding the asymptotics of the sum
similar to (\ref{eq:szp}). Kl{\o}ve in \cite{klove} studied the average worst case probability of undetected error for linear codes of length $n$ and dimension $k$ over an alphabet of size $q$
 and analyzed the following sum \\
 $\sum_{i=1}^{n}\binom{n}{i}\left(\frac{i}{n}\right)^i\left(1-\frac{i}{n}\right)^{n-i}.$
In \cite{szpa_ieee} the author obtained an asymptotics expansion of the more general sums
$\sum_{i=1}^{n}\binom{n-k}{i}\left(\frac{i}{n}\right)^i\left(1-\frac{i}{n}\right)^{n-i},$
for $k\ge 0.$ In this paper the technique used in the proof belongs to advances analytical tools.
Later Hwang in \cite{hwang} derived uniform asymptotic expressions of some Abel sums appearing in some problems in coding theory.
In \cite{spa_2013}  the authors consider the expected maximum total (i.e., sum) of movements of $n$
identical sensors placed uniformly at random in a unit interval so as to attain complete coverage of the unit interval $[0,1]$
and prove elementary the following tight asymptotic result
$$\sum_{i=1}^{n}2i\binom{n}{i}(1-t_i)^{n-i+1}(t_i)^{i}=\Theta(\sqrt{n}),\,\,\, \text{where}\,\,\, t_i=\frac{i}{n}-\frac{1}{2n}.$$
In the proof of Lemma \ref{lem:technical3} we also apply elementary techniques such as Stirling's formula (\ref{eq:stirlingform}), (\ref{eq:stirling2form}),
basic inequalities and some elementary approximation. 
\begin{Lemma}
\label{lem:technical3}
Let $c\ge 0$ and let $t_i=\frac{i}{n}-\frac{1}{2n}.$ 
Then
\begin{equation}
\label{eq:szp}
\sum_{i=1}^{n}2i\binom{n}{i}(1-t_i)^{n-i+1}(t_i)^{i+c}=n^{3/2}\frac{2}{\sqrt{2\pi}}\EBeta{c+\frac{3}{2}}{\frac{3}{2}}+O(n).
\end{equation}
\end{Lemma}
\begin{proof} 
Let $E_i=\frac{2i\left(n-i+\frac{1}{2}\right)}{n}\binom{n}{i}\left(\frac{i}{n}-\frac{1}{2n}\right)^{i+c}\left(1-\frac{i}{n}+\frac{1}{2n}\right)^{n-i}.$ We divide
the sum into four parts:
\begin{equation}
\label{eq:partition}
\sum_{i=1}^{n}E_i=\sum_{i=1}^{\lfloor \sqrt{n}\rfloor}E_i+\sum_{\lfloor \sqrt{n}\rfloor+1}^{n-\lfloor \sqrt{n}\rfloor}E_i+\sum_{n-\lfloor \sqrt{n}\rfloor +1}^{n-1}E_i+E_n
\end{equation}
We approximate the four parts separately. It is easy to see that $E_n=\Theta\left(1\right).$ For the first and third term, we use Stirling's formula (\ref{eq:stirlingform})
for $m=n,$ $m=i$ and $m=n-i,$ as well as Inequalities $e^{\frac{1}{12n}}<e,$ $e^{\frac{1}{12i+1}+\frac{1}{12(n-i)+1}}>1$
to deduce that
$$
E_i\le\frac{2e}{\sqrt{2\pi}}\frac{1}{n^{\frac{1}{2}+c}}\sqrt{(n-i)i}
\left(1+\frac{1}{2(n-i)}\right)\left(1-\frac{1}{2i}\right)^{i+c}\left(1+\frac{1}{2(n-i)}\right)^{n-i}.
$$
Applying the basic inequality $\left(1+\frac{1}{x}\right)^x<e,$ when $x\ge 1$ for $x=2(n-i)$ we have
$$E_i\le\frac{3e^{\frac{3}{2}}}{\sqrt{2\pi}}\frac{1}{n^{\frac{1}{2}+c}}\sqrt{(n-i)i}\le\frac{3e^{\frac{3}{2}}}{\sqrt{2\pi}}\frac{n^{\frac{1}{2}}}{n^{n^c}}.$$
Therefore
$$\sum_{i=1}^{\lfloor \sqrt{n}\rfloor}E_i+\sum_{n-\lfloor \sqrt{n}\rfloor +1}^{n-1}E_i=O\left(\frac{n}{n^{c}}\right)=O(n).$$
Hence the first, third and fourth term contribute $O(n)$ and the asymptotics depends on the second term.

For the second term $(\lfloor \sqrt{n}\rfloor+1\le i\le n-\lfloor \sqrt{n}\rfloor)$ we use Stirling's formula (\ref{eq:stirling2form})
for $m=n,$ $m=i$ and $m=n-i$ 
to deduce that
\begin{align*}
\binom{n}{i}&\left(\frac{i}{n}\right)^{i+c}\left(1-\frac{i}{n}\right)^{n-i}\\
&=\frac{1}{\sqrt{2\pi}\sqrt{n-i}}
\left(\frac{i}{n}\right)^{\left(-\frac{1}{2}+c\right)}\frac{1+O\left(\frac{1}{n}\right)}{\left(1+O\left(\frac{1}{i}\right)\right)\left(1+O\left(\frac{1}{n-i}\right)\right)}\\
&=\frac{1}{\sqrt{2\pi}\sqrt{n-i}}
\left(\frac{i}{n}\right)^{\left(-\frac{1}{2}+c\right)}
\left(1+O\left(\frac{1}{\sqrt{n}}\right)\right).
\end{align*}
Hence
$$
E_i=\frac{2\sqrt{n-i}}{\sqrt{2\pi}}\left(\frac{i}{n}\right)^{\frac{1}{2}+c}\left(1-\frac{1}{2i}\right)^{i+c}
\left(1+\frac{1}{2(n-i)}\right)^{n-i+1}
\left(1+O\left(\frac{1}{\sqrt{n}}\right)\right).
$$
Now we apply the approximations $\ln(1+x)=x+O(x^2),$ $e^x=1+O(x)$ and get
$\left(1-\frac{1}{2i}\right)^{i+c}\left(1+\frac{1}{2(n-i)}\right)^{n-i+1}=1+O\left(\frac{1}{\sqrt{n}}\right).$ Therefore
$$E_i=\frac{2}{\sqrt{2\pi}}\sqrt{n-i}\left(\frac{i}{n}\right)^{\frac{1}{2}+c}\left(1+O\left(\frac{1}{\sqrt{n}}\right)\right).$$
Using the inequality $\sqrt{n-i}\left(\frac{i}{n}\right)^{\frac{1}{2}+c}\le\frac{1}{2(1+c)}\left(\frac{1+2c}{2(1+c)}\right)^{\frac{1}{2}+c} \sqrt{n}$ we get
$$\sum_{i=1}^{\lfloor \sqrt{n}\rfloor}\sqrt{n-i}\left(\frac{i}{n}\right)^{\frac{1}{2}+c}
+\sum_{n-\lfloor \sqrt{n}\rfloor +1}^{n}\sqrt{n-i}\left(\frac{i}{n}\right)^{\frac{1}{2}+c}=O(n).$$
Therefore, we can add the terms back in, so we have
\begin{align*}
\sum_{\lfloor \sqrt{n}\rfloor+1}^{n-\lfloor \sqrt{n}\rfloor}E_i=&\sum_{i=1}^{n}E_i+O(n)\\
=&\frac{2}{\sqrt{2\pi}}\left(1+O\left(\frac{1}{\sqrt{n}}\right)\right)\frac{1}{n^{\frac{1}{2}+c}}\sum_{i=1}^{n}(n-i)^{\frac{1}{2}}i^{\frac{1}{2}+c}+O(n).
\end{align*}
The remaining sum we approximate with integral. Hence
$$\sum_{i=0}^{n}(n-i)^{\frac{1}{2}}i^{\frac{1}{2}+c}=\int_0^n(n-x)^{\frac{1}{2}}x^{\frac{1}{2}+c}dx+\Delta,$$
with $|\Delta|\le\sum_{i=0}^{n}\max_{i\le x<i+1}|f(x)-f(i)|$ (see \cite[Page 179]{sedgewick}.)
Observe that, the function $f(x)=x^{\frac{1}{2}+c}(n-x)^{\frac{1}{2}}$ is monotone increasing over the interval $\left[0,n\frac{1+2c}{2(1+c)}\right]$ and monotone decreasing over the interval
$\left[n\frac{1+2c}{2(1+c)},n\right].$
Hence the error term $|\Delta|$ telescopes on the interval $\left[0,\left\lfloor n\frac{1+2c}{2(1+c)}\right\rfloor \right]$ and telescopes on the interval 
$\left[\left\lceil n\frac{1+2c}{2(1+c)}\right\rceil,n\right].$
Therefore $|\Delta|=O(n^{1+c}).$
Notice that (see Equation (\ref{eq:beta}))
\begin{align*}
\int_0^n(n-x)^{\frac{1}{2}}x^{\frac{1}{2}+c}dx&=n^{2+c}\int_0^1(1-x)^{\frac{1}{2}}x^{\frac{1}{2}+c}dx\\
&=n^{2+c}\EBeta{c+\frac{3}{2}}{\frac{3}{2}}.
\end{align*}
Putting all together we deduce that the second term contributes 
\begin{align*}
&\sum_{\lfloor \sqrt{n}\rfloor+1}^{n-\lfloor \sqrt{n}\rfloor}E_i\\
&\,\,\,\,=
\frac{2}{\sqrt{2\pi}}\left(1+O\left(\frac{1}{\sqrt{n}}\right)\right)
\frac{1}{n^{\frac{1}{2}+c}}\left(n^{2+c}\EBeta{c+\frac{3}{2}}{\frac{3}{2}}+O\left(n^{1+c}\right)\right)+O(n)\\
&\,\,\,\,\,=\frac{2}{\sqrt{2\pi}}n^{\frac{3}{2}}\EBeta{c+\frac{3}{2}}{\frac{3}{2}}
+O(n).
\end{align*}
This  easily completes the proof of Lemma \ref{lem:technical3}. 
\end{proof}
Finally, we can prove the main theorem of this paper.
\begin{proof} (Theorem \ref{thm:mainexact_oddtwo1})
Assume that $a$ is an odd natural number.
Let $X_i$ be the $i-th$ order statistic, i.e., the position of the $i-th$ sensor
in the interval $[0, 1]$. It turns out (see \cite{nagaraja_1992}) that $X_i$ obeys
the Beta distribution with parameters $i, n-i+1$ and has density $i\binom{n}{i}x^{i-1}(1-x)^{n-i}$ (see Equations 
(\ref{eq:prosto}--\ref{Emult})). Let $E_i^{(a)}$ be the expected distance to the power $a$ between $X_i$ and the $i^{th}$ sensor position, $t_i=\frac{i}{n}-\frac{1}{2n},$ on the
unit interval. Therefore
$$
        E_i^{(a)}=i\binom{n}{i}\int_{0}^{1}|t_i-x|^a x^{i-1}(1-x)^{n-i}dx.
$$
To prove the asymptotic result observe that
\begin{equation}
\label{eq:ba20}
E_{i}^{(a)}=E_{i}^{(a,1)}+E_{i}^{(a,2)},
\end{equation}
where
\begin{equation}
\label{eq:latwe}
E_{i}^{(a,1)} = i\binom{n}{i}\int_{0}^{1}(x-t_i)^a x^{i-1}(1-x)^{n-i}dx,
\end{equation}
\begin{equation}
\label{eq:trudne}
E_{i}^{(a,2)} = 2i\binom{n}{i}\int_{0}^{t_i}(t_i-x)^a x^{i-1}(1-x)^{n-i}dx.
\end{equation}
The proof of Theorem \ref{thm:mainexact_oddtwo1} proceeds along the following steps.

From Lemma \ref{lem:sum1} we deduce that the sum $\sum_{i=1}^{n}E_{i}^{(a,1)}$ is negligibly and contributes
$O\left(\frac{1}{n^{\frac{a-1}{2}}}\right)$ (see Equation (\ref{first:eq})).

Then we write $E_{i}^{(a,2)}$ as the sum of $E_{i}^{(a,2,1)}$ and $E_{i}^{(a,2,2)}.$ Using Lemma \ref{lem:technical} we prove that
the sum $\sum_{i=1}^{n}E_{i}^{(a,2,1)}$ is also negligibly and contributes $O\left(\frac{1}{n^{\frac{a-1}{2}}}\right)$
(see Equation (\ref{second:eq})). 

Further, using Lemma \ref{lem:technical2} we decompose $E_{i}^{(a,2,2)}$ into the sum of 
$E_{i}^{(a,2,2,1)}$ and $F_{i}^{(a,2,2,2)}.$ 
The sum $\sum_{i=1}^{n}E_{i}^{(a,2,2,1)}$ is also neglibly
and contributes $O\left(\frac{1}{n^{\frac{a-1}{2}}}\right)$
(see Equation (\ref{eq:star4})). 

Thus the asymptotic depends on the expression given by the summand $\sum_{i=1}^{n}E_{i}^{(a,2,2,2)}$ 
Finally combining together  Lemma \ref{lem:technical2} and Lemma \ref{lem:technical3}
we deduce the main result (see Equation (\ref{eq:star5})).

We now prove the desired asymptotics.
Firstly we estimate $E_{i}^{(a,1)}.$ We show that
\begin{equation}
 \label{first:eq}
 \sum_{i=1}^n E_{i}^{(a,1)}=O\left(\frac{1}{n^{\frac{a-1}{2}}}\right).
\end{equation}
Let
$$E^{(a,1)}_{i,j}=i\binom{n}{i}{t_i}^{a-j}\binom{a}{j}(-1)^j\int_{0}^{1}x^j x^{i-1}(1-x)^{n-i}dx$$
for $j\in\{0,1,\dots ,a\}$ and $i\in\{1,2,\dots ,n\}.$ Observe that
$$E^{(a,1)}_i=\sum_{j=0}^{a}E^{(a,1)}_{i,j}.$$
The definition of the Beta function and Identity (\ref{Emult}) imply that
$$
        E^{(a,1)}_{i,j}=\left(\frac{1}{n}\right)^a\binom{a}{j}(-1)^j n^j  \frac{ \left(i-\frac{1}{2}\right)^{a-j}\cdot i^{(j)}}{(n+1)^{(j)}}.
$$
Applying Lemma \ref{lem:sum1} to the sequence $E^{(a,1)}_{i,j}$ we conclude that
$$\sum_{j=0}^{a}\sum_{i=1}^{n} E^{(a)}_{i,j}=O\left(\frac{1}{n^{\frac{a-1}{2}}}\right).$$
This finishes the proof of Equation (\ref{first:eq}).

Now we estimate $E_i^{(a,2)}.$ Let 
$$E^{(a,2)}_{i,j}=2i\binom{n}{i}{t_i}^{a-j}\binom{a}{j}(-1)^j\int_{0}^{t_i}x^j x^{i-1}(1-x)^{n-i}dx$$
for $j\in\{0,1,\dots ,a\}$ and $i\in\{1,2,\dots ,n\}.$ Observe that
$$E^{(a,2)}_i=\sum_{j=0}^{a}E^{(a,2)}_{i,j}.$$
On the other hand, Equation (\ref{Emult}) and Equation (\ref{incomplete:first}) imply that
$$\binom{n+j}{i+j}(i+j)\int_0^{t_i}x^j x^{i-1}(1-x)^{n-i}dx=I(t_i,i+j,n-i+1).$$ 
Hence from Equation (\ref{incomplete:req}) we get
\begin{align*}
I&(t_i,i+j,n-i+1)\\
&=I(t_i,i,n-i)-\sum_{k=1}^{j}\binom{n+k-1}{k+i-1}(1-t_i)^{n-i+1}(t_i)^{i+k-1}.
\end{align*}
Therefore
$$E^{(a,2)}_{i,j}=E^{(a,2,1)}_{i,j}+ E^{(a,2,2)}_{i,j},$$
where 
$$E^{(a,2,1)}_{i,j}=2\frac{i\binom{n}{i}t_i^{a-j}\binom{a}{j}(-1)j}{\binom{n+j}{i+j}(i+j)}\int_{0}^{t_i}x^{i-1}(1-x)^{n-i}dx$$
and
$$E^{(a,2,2)}_{i,j}=2i\binom{n}{i}t_i^{a-j}\binom{a}{j}(-1)^{j+1}\sum_{k=1}^j\frac{\binom{n+k-1}{k+i-1}}{\binom{n+j}{i+j}(i+j)}
(1-t_i)^{n-i+1}(t_i)^{i+k-1}.$$
Let $E^{(a,2,1)}_{i}=\sum_{j=0}^aE^{(a,2,1)}_{i,j}$ and 
$E^{(a,2,2)}_{i}=\sum_{j=0}^aE^{(a,2,2)}_{i,j}.$ Hence
\begin{equation}
\label{eq:a20}
E^{(a,2)}_{i}=E^{(a,2,1)}_{i}+E^{(a,2,2)}_{i}.
\end{equation}
Using Lemma \ref{lem:technical} we get
\begin{equation}
  \label{second:eq}
 \sum_{i=1}^{  n }E_{i}^{(a,2,1)}=O\left(\frac{1}{n^{\frac{a-1}{2}}}\right).
\end{equation}
Observe that
$$
E_{i}^{(a,2,2)}=2i\binom{n}{i}(1-t_i)^{n-i+1} (t_i)^{i-1+a}
\sum_{j=1}^{a}\binom{a}{j}(-1)^{j+1}\sum_{k=1}^{j}\frac{\binom{n+k-1}{k+i-1}}{\binom{n+j}{i+j}(i+j)}(t_i)^{k-j}.
$$

Let
\begin{align*}
&E_{i}^{(a,2,2,1)}\\
&\,\,\,\,=2i\binom{n}{i}(1-t_i)^{n-i+1}(t_i)^{i-1+a}\frac{1}{(n+1)^{\overline{a}}}
\sum_{
S_5
}
b_{q_1, p_1}(a)n^{a-1-q_1}\left(i-\frac{1}{2}\right)^{-p_1},
\end{align*}
\begin{align*}
&E_{i}^{(a,2,2,2)}\\
&\,\,\,\,=2i\binom{n}{i}(1-t_i)^{n-i+1}(t_i)^{i-1+a}\frac{1}{(n+1)^{\overline{a}}}
\sum_{
S_6
}
b_{q_1, p_1}(a)n^{a-1-q_1}\left(i-\frac{1}{2}\right)^{-p_1},
\end{align*}
where
$$
S_5=\left\{(q_1,p_1): 0\le q_1,p_1\le a-1, q_1+p_1> \frac{a-1}{2} \right\},
$$
$$
S_6=\left\{(q_1,p_1): 
0\le q_1, p_1\le a-1, q_1+p_1= \frac{a-1}{2}
\right\}.
$$
Applying Lemma \ref{lem:technical2} we deduce that
\begin{equation}
\label{eq:a24}
E_{i}^{(a,2,2)}=E_{i}^{(a,2,2,1)}+E_{i}^{(a,2,2,2)}.
\end{equation}
From Lemma \ref{lem:technical3} for $c=a-1-p_1$ we have
\begin{equation}
\label{eq:star4}
\sum_{i=1}^{n}E_{i}^{(a,2,2,1)}=O\left(\frac{1}{n^{\frac{a}{2}}}\right).
\end{equation}
From Lemma \ref{lem:technical3} for $c=a-1-p_1$ and
Equation (\ref{eq:technical2b}) in Lemma \ref{lem:technical2}
we deduce that
\begin{equation}
\label{eq:star5}
\sum_{i=1}^{n}E_{i}^{(a,2,2,2)}=\frac{\Gamma\left(\frac{a}{2}+1\right)}{2^{\frac{a}{2}}(1+a)}\frac{1}{n^{\frac{a}{2}-1}}+O\left(\frac{1}{n^{\frac{a-1}{2}}}\right).
\end{equation}
Finally,  putting together 
Equations (\ref{eq:ba20}--\ref{eq:star5})
finishes the proof of Theorem \ref{thm:mainexact_oddtwo1}. 
\end{proof}
\bibliographystyle{plain}
\bibliography{refs}

\begin{thebibliography}{10}

\bibitem{nagaraja_1992}
B.~Arnold, N.~Balakrishnan, and H.~Nagaraja.
\newblock {\em A first course in order statistics}, volume~54.
\newblock SIAM, 1992.

\bibitem{feller1968}
W.~Feller.
\newblock {\em An Introduction to Probability Theory and its Applications},
  volume~1.
\newblock John Wiley, NY, 1968.

\bibitem{sedgewick}
P.~Flajolet and B.~Sedgewick.
\newblock {\em An Introduction to the Analysis of Algorithms}.
\newblock Addison-Wesley, 1995.

\bibitem{Gould}
H.~Gould and J.~Quaintance.
\newblock {\em Combinatorial Identities for Stirling Numbers}.
\newblock World Scientific Publishing Co., Singapore, 2015.

\bibitem{concrete_1994}
R.~Graham, D.~Knuth, and O.~Patashnik.
\newblock {\em Concrete Mathematics A Foundation for Computer Science}.
\newblock Addison-Wesley, Reading, MA, 1994.

\bibitem{hwang}
H.~Hwang.
\newblock Uniform asymptotics of some abel sums arising in coding theory.
\newblock {\em Theoretical Computer Science}, 263:145--158, 2001.

\bibitem{kapelkokranakisIPL}
R.~Kapelko and E.~Kranakis.
\newblock On the displacement for covering a unit interval with randomly placed
  sensors.
\newblock {\em Information Processing Letters}, 116:710--717, 2016.

\bibitem{klove}
T.~Kl{\o}ve.
\newblock Bounds on the worst case probability of undetected error.
\newblock {\em IEEE Transactions on Information Theory}, 41(1):298--300, 1995.

\bibitem{spa_2013}
E.~Kranakis, D.~Krizanc, O.~Morales-Ponce, L.~Narayanan, J.~Opatrny, and
  S.~Shende.
\newblock Expected sum and maximum of displacement of random sensors for
  coverage of a domain.
\newblock In {\em Proceedings of the 25th ACM symposium on Parallelism in
  algorithms and architectures}, pages 73--82. ACM, 2013.

\bibitem{NIST}
NIST Digital~Library of~Mathematical~Functions.
\newblock http://dlmf.nist.gov/8.17.

\bibitem{szpa_ieee}
W.~Szpankowski.
\newblock On asymptotics of certain sums arising in coding theory.
\newblock {\em IEEE Transactions on Information Theory}, 41(6):2087--2090,
  1995.

\end{thebibliography}

\end{document}